\documentclass[envcountersame,runningheads]{llncs}

\usepackage[T1]{fontenc}
\usepackage{booktabs}
\usepackage{hyperref}
\usepackage{xcolor}
\hypersetup{colorlinks=true,citecolor=blue,urlcolor=blue,linkcolor=blue}
\usepackage{xspace}
\usepackage{scalerel}
\usepackage{extpfeil}
\usepackage{stackrel}
\usepackage[misc]{ifsym}

\newcommand{\m}[1]{\mathsf{#1}}
\newcommand{\mr}[1]{\mathrel{#1}}
\newcommand{\RB}[2][1]{\raisebox{#1mm}{#2}}
\newcommand{\RBX}[2][1]{\raisebox{#1ex}{#2}}

\newcommand{\xF}{\mathcal{F}}
\newcommand{\xT}{\mathcal{T}}
\newcommand{\xV}{\mathcal{V}}
\newcommand{\xI}{\mathcal{I}}
\newcommand{\xJ}{\mathcal{J}}

\newcommand{\Val}{\xV\m{al}}
\newcommand{\Var}{\xV\m{ar}}
\newcommand{\LVar}{\mathcal{L}\Var}
\newcommand{\EVar}{\mathcal{E}\Var}
\newcommand{\Dom}{\mathcal{D}\m{om}}
\newcommand{\xFTe}{\xF_{\m{terms}}}
\renewcommand{\xFTe}{\xF_{\m{te}}}
\newcommand{\xFTh}{\xF_{\m{theory}}}
\renewcommand{\xFTh}{\xF_{\m{th}}}
\newcommand{\inter}[1]{[\![{#1}]\!]}
\newcommand{\overlap}[3][p]{\langle #2, #1, #3 \rangle}
\newcommand{\Pos}{\mathcal{P}\m{os}}
\newcommand{\VPOS}{\Pos_\xV}
\newcommand{\FPOS}{\Pos_\xF}
\newcommand{\CCP}{\m{CCP}}
\newcommand{\tf}{\m{tf}}

\newcommand{\hole}{\square}
\newcommand\overlaps{\blacktriangle}
\newcommand{\Superterm}{\mr{{\vartriangleright}}}
\newcommand\Supertermmul{\Superterm_\m{mul}}
\newcommand\Supertermmuleq{\Superterm_\m{mul}^{=}}

\newcommand{\R}{\rightarrow}
\newcommand{\Ra}[1][]{\R^{#1}}
\newcommand{\Rb}[1][]{\R_{#1}}
\newcommand{\Rab}[2][]{\R_{#1}^{#2}}
\newcommand{\RbR}{\Rb[\xR]}
\newcommand{\xR}{\mathcal{R}}
\newcommand{\xRa}[1][]{\xrightarrow{#1}}
\newcommand{\RabR}[1]{\Rab[\xR]{#1}}
\newcommand{\LabR}[1]{\mr{\vphantom{\R}_{\xR}^{#1}{\L}}}
\newcommand{\LabRs}[1]{\mr{\prescript{*}{\xR}{\Ls}}}
\renewcommand{\L}{\leftarrow}
\newcommand{\La}[1][]{\mr{\vphantom{\R}^{#1}{\L}}}
\newcommand{\Lb}[1][]{\mr{\vphantom{\R}_{#1}{\L}}}
\newcommand{\LbR}{\Lb[\xR]}
\newcommand{\xLRb}[1][]{\xleftrightarrow[#1]{}}
\newcommand{\Rs}{\stackrel{\smash{\RB[-.5]{\tiny $\sim$~}}}{\R}}
\newcommand{\Ls}{\stackrel{\smash{\RB[-.5]{\tiny ~$\sim$}}}{\L}}
\newcommand{\Rbs}[1][]{\Rs_#1}
\newcommand{\RbRs}{\Rs_\xR}
\newcommand{\RabRs}[1]{\Rs_\xR^{#1}}

\newcommand{\J}{\downarrow}
\newcommand{\Jb}[1][]{\J_{#1}}
\newcommand{\JbR}{\Jb[\xR]}

\newcommand{\smallparallel}{\raisebox{.07em}{\scalebox{.6}{$\|$}}}
\newcommand{\PR}{\mr{\smash{\xrightarrow%
[\,\smash{\RBX[1.45]{\smallparallel}}\,]{}}}}
\newcommand{\PRa}[1][]{\PR^{#1}}
\newcommand{\PRb}[1][]{\PR_{#1}}
\newcommand{\PRbR}{\PRb[\xR]}
\newcommand{\PRs}{\stackrel{\smash{\RB[1.2]{\tiny $\sim$\,}}}{\PR}}
\newcommand{\PRbs}[1][]{\PRs_{#1}}

\newcommand{\rPR}{\mr{\smash{\xleftarrow%
[\,\smash{\RBX[1.45]{\smallparallel}}\,]{}}}}
\newcommand{\rPRb}[1][]{\mr{\vphantom{\PR}_{#1}\smash{\xleftarrow%
[\,\smash{\RBX[1.45]{\smallparallel}}\,]{}}}}
\newcommand{\rPRa}[1][]{\mr{\vphantom{\PR}^{#1}\smash{\xleftarrow%
[\,\smash{\RBX[1.45]{\smallparallel}}\,]{}}}} 

\newcommand{\corref}[1]{Corollary~\ref{cor:#1}}
\newcommand{\secref}[1]{Section~\ref{sec:#1}}
\newcommand{\defref}[1]{Definition~\ref{def:#1}}

\newcommand{\tabref}[1]{Table~\ref{tab:#1}}
\newcommand{\lemref}[1]{Lemma~\ref{lem:#1}}
\newcommand{\thmref}[1]{Theorem~\ref{thm:#1}}
\newcommand{\exaref}[1]{Example~\ref{exa:#1}}

\newcommand{\Rca}{\Rb[\m{ca}]}
\newcommand{\xRca}{\xR_\m{ca}}
\newcommand{\Rru}{\Rb[\m{ru}]}

\newcommand{\crest}{\textsf{crest}\xspace}
\newcommand{\ctrl}{\textsf{Ctrl}\xspace}
\newcommand{\seq}[2][n]{{#2_1},\dots,{#2_{#1}}}
\newcommand{\h}[1][.3]{\hspace{#1mm}}
\newcommand{\sh}{\h|\h}
\newcommand{\SET}[1]{\{\h#1\h\}}
\newcommand{\CO}[1]{[\h#1\h]}
\newcommand{\BB}{\mathbb{B}}

\newcommand{\ZZ}{\mathbb{Z}}
\newcommand{\crr}[3]{#1 \R #2~\CO{#3}}
\newcommand{\CRR}{\crr{\ell}{r}{\varphi}}
\newcommand{\ML}[2][0]{\makebox[#1mm][l]{#2}}

\begin{document}

\title{Confluence Criteria for Logically Constrained Rewrite Systems%
\thanks{This research is supported by FWF (Austrian Science Fund) project
I 5943-N.}}

\author{Jonas Sch\"opf\orcidID{0000-0001-5908-8519}%
\textsuperscript{(\href{mailto:jonas.schoepf@uibk.ac.at}{\Letter})} \and
Aart Middeldorp\orcidID{0000-0001-7366-8464}}
\institute{Department of Computer Science, Universit\"at Innsbruck,
Innsbruck, Austria \\
\email{\{jonas.schoepf,aart.middeldorp\}@uibk.ac.at}}

\authorrunning{Jonas Sch\"opf \and Aart Middeldorp}

\maketitle

\begin{abstract}
Numerous confluence criteria for plain term rewrite systems are
known. For logically constrained rewrite system, an attractive extension
of term rewriting in which rules are equipped with logical constraints,
much less is known. In
this paper we extend the strongly-closed and (almost) parallel-closed
critical pair criteria of Huet and Toyama to the logically
constrained setting. We discuss the challenges for automation and present
\crest, a new tool for logically constrained rewriting in which the
confluence criteria are implemented, together with experimental data.
\end{abstract}

\section{Introduction}
\label{sec:introduction}

Logically constrained rewrite systems constitute a general rewrite
formalism with native support for constraints that are handled by SMT
solvers. They are useful for program analysis, as illustrated in numerous
papers~\cite{CL18,FKN17,KN23,WM21}. Several results from term
rewriting have been lifted to constrained rewriting. We mention
termination analysis~\cite{K13,K16,WM18}, rewriting induction~\cite{FKN17},
completion~\cite{WM18} as well as runtime complexity analysis~\cite{WM21}.

In this paper we are concerned with confluence analysis of
logically constrained rewrite systems (LCTRSs for short).
Only two sufficient conditions for confluence of LCTRSs
are known. Kop and Nishida considered (weak) orthogonality
in \cite{KN13}. Orthogonality is the combination of left-linearity and
the absence of critical pairs, in a weakly orthogonal system trivial
critical pairs are allowed. Completion of LCTRSs
is the topic of \cite{WM18} and the underlying confluence condition
of completion is the combination of termination and joinability of
criticial pairs. In this paper we add two further confluence criteria.
Both of these extend known conditions for standard term rewriting to
the constrained setting. The first is the combination of linearity and
strong closedness of critical pairs, introduced by Huet~\cite{H80}.
The second, also due to \cite{H80}, is the combination of left-linearity
and parallel closedness of critical pairs. We also consider an
extension of the latter, due to Toyama~\cite{T88}.

\paragraph{Overview}
The remainder of this paper is organized as follows. In the next section
we summarize the relevant background.
\secref{confluence} recalls the existing confluence criteria
for LCTRSs and some of the underlying results. The new confluence
criteria for LCTRSs are reported in \secref{results}.
In \secref{automation} the automation challenges we faced are
described and we present our prototype implementation \crest.
Experimental results are reported in \secref{experiments},
before we conclude in \secref{conclusion}.

\section{Preliminaries}
\label{sec:prelims}

We assume familiarity with the basic notions of term rewrite systems
(TRSs)~\cite{BN98}, but shortly recapitulate terminology and
notation that we use in the remainder. In particular, we recall the 
notion of logically constrained rewriting as defined in~\cite{FKN17,KN13}.

We assume a many-sorted signature $\xF$ and
a set $\xV$ of (many-sorted) variables disjoint from $\xF$.
The signature $\xF$ is split into term symbols
from $\xFTe$ and theory symbols from $\xFTh$. The set $\xT(\xF,\xV)$
contains the well-sorted terms over this signature and
$\xT(\xFTh)$ denotes the set of well-sorted ground terms that consist
entirely of theory symbols.
We assume a mapping $\xI$ which assigns to every sort $\iota$ occurring
in $\xFTh$ a carrier set $\xI(\iota)$, and an interpretation
$\xJ$ that assigns to every symbol $f \in \xFTh$ with sort
declaration $\iota_1 \times \cdots \times \iota_n \to \kappa$ a function
$f_\xJ\colon \xI(\iota_1) \times \cdots \times \xI(\iota_n) \to
\xI(\kappa)$. Moreover, for every sort $\iota$ occurring in $\xFTh$ we
assume a set
$\Val_\iota \subseteq \xFTh$ of value symbols, such that all
$c \in \Val_\iota$ are constants of sort $\iota$ and $\xJ$ constitutes a
bijective mapping between $\Val_\iota$ and $\xI(\iota)$. Thus there exists
a constant symbol in $\xFTh$ for every value in the carrier set.
The interpretation $\xJ$ naturally extends to a mapping
$\inter{\cdot}$ from ground terms in $\xT(\xFTh)$ to values in
$\Val = \bigcup_{\iota \in \Dom(\xI)} \Val_{\iota}$:
$\inter{f(\seq{t})} = f_\xJ(\inter{t_1},\dots,\inter{t_n})$
for all $f(\seq{t}) \in \xT(\xFTh)$.
So every ground term in $\xT(\xFTh)$ has a unique value.
We demand that theory symbols and term symbols overlap only on values,
i.e., $\xFTe \cap \xFTh \subseteq \Val$. A term in $\xT(\xFTh,\xV)$
is called a \emph{logical} term.

Positions are strings of positive natural numbers used to address
subterms. The empty string is denoted by $\epsilon$. We write
$q \leqslant p$ and say that $p$ is below $q$ if $qq' = p$ for some
position $q'$, in which case $p \backslash q$ is defined to be $q'$.
Furthermore, $q < p$ if $q \leqslant p$ and $q \neq p$. Finally, positions
$q$ and $p$ are parallel, written as $q \parallel p$, if neither
$q \leqslant p$ nor $p < q$.
The set of positions of a term $t$ is defined as
$\Pos(t) = \SET{\epsilon}$ if $t$ is a variable or a constant,
and as $\Pos(t) = \SET{\epsilon} \cup \SET{iq \mid
\text{$1 \leqslant i \leqslant n$ and $q \in \Pos(t_i)$}}$ if
$t = f(\seq{t})$ with $n \geqslant 1$. The subterm of $t$ at position
$p \in \Pos(t)$ is defined as $t|_p = t$ if $p = \epsilon$ and as
$t|_p = t_i|_q$ if $p = iq$ and $t = f(\seq{t})$. We write $s[t]_p$ for
the result of replacing the subterm at position $p$ of $s$ with $t$.
We write $\VPOS(t)$ for $\SET{p \in \Pos(t) \mid t|_p \in \xV}$ and
$\FPOS(t)$ for $\Pos(t) \setminus \VPOS(t)$.
The set of variables occurring in the term $t$ is denoted by
$\Var(t)$. A term $t$ is linear if every variable occurs at most once in
it. A substitution is a mapping $\sigma$ from $\xV$ to $\xT(\xF,\xV)$ such
that its domain $\SET{x \in \xV \mid \sigma(x) \neq x}$ is finite. We
write $t\sigma$ for the result of applying $\sigma$ to the term $t$.

We assume the existence of a sort $\m{bool}$ such that
$\xI(\m{bool}) = \BB = \SET{\top,\bot}$,
$\Val_{\m{bool}} = \SET{\m{true},\m{false}}$, $\inter{\m{true}} = \top$,
and $\inter{\m{false}} = \bot$ hold.
Logical terms of sort $\m{bool}$ are called \emph{constraints}.
A constraint $\varphi$ is \emph{valid} if
$\inter{\varphi\gamma} = \top$ for all substitutions $\gamma$ such that
$\gamma(x) \in \Val$ for all $x \in \Var(\varphi)$.

A \emph{constrained rewrite rule} is a triple $\CRR$ where
$\ell, r \in \xT(\xF,\xV)$ are terms of the same sort such that
$\m{root}(\ell) \in \xFTe \setminus \xFTh$ and $\varphi$
is a logical term of sort $\m{bool}$. If $\varphi = \m{true}$ then the
constraint is often omitted, and the rule is denoted as $\ell \R r$.
We denote the set $\Var(\varphi) \cup (\Var(r) \setminus \Var(\ell))$ of
\emph{logical} variables in $\rho\colon \CRR$ by
$\LVar(\rho)$.
We write $\EVar(\rho)$ for the set
$\Var(r) \setminus (\Var(\ell) \cup \Var(\varphi))$ of
variables that appear only in the right-hand side of $\rho$.
Note that extra variables in right-hand sides are allowed, but they
may only be instantiated by values. This is useful to model
user input or random choice~\cite{FKN17}.
A set of constrained rewrite rules
is called a \emph{logically constrained rewrite system}
(LCTRS for short).

The LCTRS $\xR$ introduced in the example below computes the
maximum of two integers.

\begin{example}
\label{exa:example1}
Before giving the rules, we need to define the term and theory symbols,
the carrier sets and interpretation functions:
\begin{align*}
\xFTe \,=\, {}&\SET{\m{max} \colon \m{int} \times \m{int} \Rightarrow
\m{int}} \cup \SET{\m{0}, \m{1}, \dots \colon \m{int}}
\qquad \xI_{\m{bool}} \,=\, \BB \qquad \xI_{\m{int}} \,=\, \ZZ \\
\xFTh \,=\, {}&\SET{\m{0}, \m{1}, \dots \colon \m{int}} \cup
\SET{\m{true}, \m{false} \colon \m{bool}} \cup
\SET{\lnot \colon \m{bool} \Rightarrow \m{bool}} \\
{} \,\cup\, {}&\SET{- \colon \m{int} \Rightarrow \m{int}} \cup
\SET{\land \colon \m{bool} \times \m{bool} \Rightarrow \m{bool}} \\
{} \,\cup\, {}&\SET{+, - \colon \m{int} \times \m{int} \Rightarrow
\m{int}} \cup \SET{\leq, \geq, <, >, =\,\colon \m{int} \times \m{int}
\Rightarrow \m{bool}}
\end{align*}
The interpretations for theory symbols follow the usual semantics given
in the SMT-LIB theory
\textsf{Ints}\footnote{\url{http://smtlib.cs.uiowa.edu/Theories/Ints.smt2}}
used by the SMT-LIB logic \textsf{QF\_LIA}. The LCTRS $\xR$ consists of
the following constrained rewrite rules
\begin{align*}
\m{max}(x,y) &\R x ~ \CO{x \geq y} &
\m{max}(x,y) &\R y ~ \CO{y \geq x} &
\m{max}(x,y) &\R \m{max}(y,x)
\end{align*}
\end{example}

In later examples we refrain from spelling out
the signature and interpretations of the theory \textsf{Ints}.
We now define rewriting using constrained rewrite rules.
LCTRSs admit two kinds of rewrite steps. Rewrite rules give rise to
\emph{rule} steps, provided the constraint of the
rule is satisfied. In addition, theory calls of the form
$f(\seq{v})$ with $f \in \xFTh \setminus \Val$ and
values $\seq{v}$ can be evaluated in a \emph{calculation} step.
In the definition below, a
substitution $\sigma$ is said to \emph{respect} a rule
$\rho\colon \CRR$, denoted by $\sigma \vDash \rho$, if
$\Dom(\sigma) = \Var(\ell) \cup \Var(r) \cup \Var(\varphi)$,
$\sigma(x) \in \Val$ for all $x \in \LVar(\rho)$, and
$\varphi\sigma$ is
valid. Moreover, a constraint $\varphi$ is respected by $\sigma$,
denoted by $\sigma \vDash \varphi$, if $\sigma(x) \in \Val$ for all
$x \in \Var(\varphi)$ and $\varphi\sigma$ is valid.

\begin{definition}
Let $\xR$ be an LCTRS. A \emph{rule step} $s \Rru t$ satisfies
$s|_p = \ell\sigma$ and $t = s[r\sigma]_p$ for some
position $p$ and
constrained rewrite rule $\CRR$ that is respected by the substitution
$\sigma$. A \emph{calculation step} $s \Rca t$ satisfies
$s|_p = f(\seq{v})$ and $t = s[v]_p$ for some
$f \in \xFTh \setminus \Val$,
$\seq{v} \in \Val$ with $v = \inter{f(\seq{v})}$.
In this case $f(\seq{x}) \R y~\CO{y = f(\seq{x})}$
with a fresh variable $y$ is a \emph{calculation rule}. The set of all
calculation rules is denoted by $\xRca$.
The relation $\RbR$ associated with $\xR$ is the union of
$\Rru \cup \Rca$.
\end{definition}

We sometimes write $\Rb[p\sh\rho\sh\sigma]$ to indicate that the
rewrite step takes place at position $p$, using the constrained rewrite
rule $\rho$ with substitution $\sigma$.

\begin{example}
We have
$\m{max}(\m{1} + \m{2}, \m{4}) \RbR \m{max}(\m{3}, \m{4}) \RbR
\m{max}(\m{4},\m{3}) \RbR \m{4}$
in the LCTRS of \exaref{example1}. The first step is a calculation step.
In the third step we apply the rule $\m{max}(x,y) \R x~\CO{x \geq y}$
with substitution $\sigma = \SET{x \mapsto \m{4}, y \mapsto \m{3}}$.
\end{example}

\section{Confluence}
\label{sec:confluence}

In this paper we are concerned with the confluence of LCTRSs. An
LCTRS $\xR$ is \emph{confluent} if $t \RabR{*} \cdot \LabR{*} u$
for all terms $s$, $t$ and $u$ such that $t \LabR{*} s \RabR{*} u$.
Confluence criteria for TRSs are based on critical pairs.
Critical pairs for LCTRS were introduced in \cite{KN13}. The
difference with the definition below is that we add dummy constraints
for \emph{extra} variables in right-hand sides of rewrite rules.

\begin{definition}
\label{def:ccp}
An \emph{overlap} of an LCTRS $\xR$ is a triple $\overlap{\rho_1}{\rho_2}$
with rules $\rho_1\colon \crr{\ell_1}{r_1}{\varphi_1}$ and
$\rho_2\colon \crr{\ell_2}{r_2}{\varphi_2}$,
satisfying the following conditions:
\begin{enumerate}
\item
$\rho_1$ and $\rho_2$
are variable-disjoint variants of rewrite rules in $\xR \cup \xRca$,
\item
$p \in \FPOS(\ell_2)$,
\item
$\ell_1$ and $\ell_2|_p$ are unifiable with a mgu $\sigma$ such that
$\sigma(x) \in \Val \cup \xV$ for all
$x \in \LVar(\rho_1) \cup \LVar(\rho_2)$,
\item
$\varphi_1\sigma \land \varphi_2\sigma$ is satisfiable, and
\item
if $p = \epsilon$ then $\rho_1$ and $\rho_2$ are not variants, or
$\Var(r_1) \nsubseteq \Var(\ell_1)$.
\end{enumerate}
In this case we call
\(
\ell_2\sigma[r_1\sigma]_p \approx r_2\sigma~
\CO{\varphi_1\sigma \land \varphi_2\sigma \land \psi\sigma}
\)
a \emph{constrained critical pair} obtained from the overlap
$\overlap{\rho_1}{\rho_2}$.
Here
\[
\psi = \bigwedge~\SET{x = x \mid x \in \EVar(\rho_1) \cup \EVar(\rho_2)}
\]
The set of all constrained critical pairs of $\xR$ is denoted by
$\CCP(\xR)$.
\end{definition}

In the following we drop ``constrained'' and speak of critical pairs.
The condition $\Var(r_1) \nsubseteq \Var(\ell_1)$ in the fifth condition
is essential to correctly deal with extra variables in rewrite rules.
The equations ($\psi$) added to the constraint of a
critical pair save the information which variables in a critical
pair were introduced by variables only occurring in the right-hand
side of a rewrite rule and therefore should
\emph{only} be instantiated by values. 
Critical pairs as defined in~\cite{KN13,WM18} lack this information.
The proof of \thmref{sccp} in the next section makes clear why those
trivial equations are essential for our confluence criteria, see
also \exaref{new}.

\begin{example}
Consider the LCTRS consisting of the rule
\[
\rho\colon ~ \m{f}(x) \R z~\CO{x = z \textbf{\textasciicircum} \m{2}}
\]
The variable $z$ does not occur in the left-hand side and the condition
$\Var(r_1) \nsubseteq \Var(\ell_1)$ ensures that $\rho$ overlaps with (a
variant of) itself at the root position. Note that $\xR$ is not confluent
due to the non-joinable local peak $\m{-4} \L \m{f}(\m{16}) \R \m{4}$.
\end{example}

\begin{example}
\label{exa:example1 ccp}
The LCTRS $\xR$ of \exaref{example1} admits the following critical pairs:
\begin{align*}
x &\approx y~\CO{x \geq y \land y \geq x} &&\overlap[\epsilon]{1}{2} \\
x &\approx \m{max}(y,x)~\CO{x \geq y} &&\overlap[\epsilon]{1}{3} \\
y &\approx \m{max}(y,x)~\CO{y \geq x} &&\overlap[\epsilon]{2}{3}
\end{align*}
The originating overlap is given on the right, where we number the
rewrite rules from left to right in \exaref{example1}.
\end{example}

Actually, there are three more overlaps since the position of overlap
($\epsilon$) is the root position. Such overlaps are called
\emph{overlays} and always come in pairs. For instance,
$\m{max}(y,x) \approx x~\CO{x \geq y}$ is the critial
pair originating from $\overlap[\epsilon]{3}{1}$. For confluence criteria
based on symmetric joinability conditions of critical pairs (like weak
orthogonality and joinability of critical pairs for terminating systems)
we need to consider just one critical pair, but this is not true
for the criteria presented in the next section.

Logically constrained rewriting aims to rewrite
(unconstrained) terms with constrained rules. However, for
the sake of analysis, rewriting \emph{constrained terms}
is useful. In particular, since critical pairs in LCTRSs come
with a constraint, confluence criteria need to consider
constrained terms. The relevant notions defined below
originate from~\cite{FKN17,KN13}.

\begin{definition}
\label{def:constrained rewriting}
A \emph{constrained term} is a pair $s~\CO{\varphi}$ of a term $s$ and
a constraint $\varphi$. Two constrained terms $s~\CO{\varphi}$ and
$t~\CO{\psi}$ are \emph{equivalent},
denoted by $s~\CO{\varphi} \sim t~\CO{\psi}$, if for every substitution
$\gamma$ respecting $\varphi$ there is some substitution $\delta$ that
respects $\psi$ such that $s\gamma = t\delta$, and vice versa.
Let $\xR$ be an LCTRS and $s~\CO{\varphi}$ a constrained term. If
$s|_p = \ell\sigma$ for some constrained rewrite rule
$\rho\colon \crr{\ell}{r}{\psi}$, position $p$, and substitution
$\sigma$ such that $\sigma(x) \in \Val \cup \Var(\varphi)$ for all
$x \in \LVar(\rho)$, $\varphi$ is satisfiable and
$\varphi \Rightarrow \psi\sigma$ is valid then
\[
s~\CO{\varphi} \Rru s[r\sigma]_p~\CO{\varphi}
\]
is a \emph{rule step}. If $s|_p = f(\seq{s})$ with
$f \in \xFTh \setminus \Val$ and $\seq{s} \in \Val \cup \Var(\varphi)$
then
\[
s~\CO{\varphi} \Rca s[x]_p~\CO{\varphi \land x = f(\seq{s})}
\]
is a \emph{calculation step}. Here $x$ is a fresh variable.
We write $\RbR$ for ${\Rru} \cup {\Rca}$ and the rewrite relation
$\RbRs$ on constrained terms is defined as $\sim \cdot \RbR \cdot \sim$.
\end{definition}

Positions in connection with $\RbRs$ steps always refer to the
underlying steps in $\RbR$. We give an example of constrained
rewriting.

\begin{example}
Consider again the LCTRS $\xR$ of \exaref{example1}. We have
\begin{align*}
\m{max}(x + y, \m{6}) ~ \CO{x \geq \m{2} \land y \geq \m{4}}
&\RbR \m{max}(z, \m{6}) ~
\CO{x \geq \m{2} \land y \geq \m{4} \land z = x + y} \\
&\RbR z ~ \CO{x \geq \m{2} \land y \geq \m{4} \land z = x + y}
\end{align*}
The first step is a calculation step. The second step is a rule step
using the rule $\m{max}(x,y) \R x ~ \CO{x \geq y}$ with the
substitution $\sigma = \SET{x \mapsto z, y \mapsto \m{6}}$.
Note that the constraint
$(x \geq \m{2} \land y \geq \m{4} \land z = x + y) \Rightarrow
z \geq \m{6}$ is valid.
\end{example}

\begin{definition}
A critical pair $s \approx t~\CO{\varphi}$ is \emph{trivial} if
$s\sigma = t\sigma$ for every substitution $\sigma$ with
$\sigma \vDash \varphi$.%
\footnote{The triviality condition in \cite{KN13} is wrong. Here we use
the corrected version in an update of \cite{KN13} announced on
Cynthia Kop's website
(accessible at \url{https://www.cs.ru.nl/~cynthiakop/frocos13.pdf}).}
A left-linear LCTRS having only trivial critical pairs is called
\emph{weakly orthogonal}. A left-linear TRS without critical pairs
is called \emph{orthogonal}.
\end{definition}

The following result is from~\cite{KN13}.

\begin{theorem}
\label{thm:wo}
Weakly orthogonal LCTRS are confluent.
\qed
\end{theorem}

\begin{example}
\label{exa:example6}
The following left-linear LCTRS computes the Ackermann function using
term symbols from
$\xFTe = \SET{\m{ack} : \m{int} \times \m{int} \Rightarrow \m{int}} \cup
\SET{\m{0}, \m{1}, \dots : \m{int}}$ and the same theory symbols,
carrier sets and interpretations as in~\exaref{example1}:
\begin{align*}
\m{ack}(\m{0},n) &\R n + \m{1} ~ \CO{n \geq \m{0}} \\
\m{ack}(m,\m{0}) &\R \m{ack}(m - \m{1},\m{1}) ~ \CO{m > \m{0}} \\
\m{ack}(m,n) &\R \m{ack}(m - \m{1},\m{ack}(m,n - \m{1})) ~
\CO{m > \m{0} \land n > \m{0}} \\
\m{ack}(m,n) &\R \m{0} ~ \CO{m < \m{0} \lor n < \m{0}}
\end{align*}
Since the conjunction of any two constraints is unsatisfiable, $\xR$
lacks critical pairs. Hence $\xR$ is confluent by \thmref{wo}.
\end{example}

The following result is proved in \cite{WM18} and forms the basis
of completion of LCTRSs.

\begin{lemma}
\label{lem:cpl}
Let $\xR$ be an LCTRS. If $t \LbR s \RbR u$ then $t \JbR u$
or $\smash{t \xLRb[\CCP(\xR)] u}$.
\qed
\end{lemma}

In combination with Newman's Lemma, the following confluence criterion
is obtained.

\begin{corollary}
\label{cor:jcp}
A terminating LCTRS is confluent if all critical pairs are joinable.
\end{corollary}

This is less obvious than it seems. Joinability of a critical pair
$s \approx t~\CO{\varphi}$ cannot simply be defined as
$s~\CO{\varphi} \RabRs{*} \cdot \LabRs{*} t~\CO{\varphi}$,
as the following example shows.

\begin{example}
\label{exa:example2}
Consider the terminating LCTRS $\xR$ consisting of the rewrite rules
\begin{align*}
\m{f}(x,y) &\R \m{g}(x,\m{1} + \m{1}) &
\m{h}(\m{f}(x,y)) &\R \m{h}(\m{g}(y,\m{1} + \m{1}))
\end{align*}
The single critical pair
$\m{h}(\m{g}(x,\m{1} + \m{1})) \approx \m{h}(\m{g}(y,\m{1} + \m{1}))$
should not be joinable because $\xR$ is not confluent, but we do have
\begin{align*}
\m{h}(\m{g}(x,\m{1} + \m{1}))
&\Rca \m{h}(\m{g}(x,z))~\CO{z = \m{1} + \m{1}} 
\sim \m{h}(\m{g}(y,v))~\CO{v = \m{1} + \m{1}} \\
\m{h}(\m{g}(y,\m{1} + \m{1}))
&\Rca \m{h}(\m{g}(y,v))~\CO{v = \m{1} + \m{1}}
\end{align*}
due to the equivalence relation $\sim$ on constrained terms; since
$x$ and $y$ do not appear in the constraints, there is no demand that
they must be instantiated with values.
\end{example}

The solution is not to treat the two sides of a critical pair in
isolation but define joinability based on rewriting constrained
term pairs. So we view the symbol $\approx$ in a constrained equation
$s \approx t~\CO{\varphi}$ as a binary constructor symbol such that
the constrained equation can be viewed as a constrained term. Steps in
$s$ take place at positions $\geqslant 1$ whereas steps in $t$
use positions $\geqslant 2$. The same is done in completion of
LCTRSs~\cite{WM18}.

\begin{definition}
We call a constrained equation $s \approx t~\CO{\varphi}$ \emph{trivial}
if $s\sigma = t\sigma$ for any substitution $\sigma$ with
$\sigma \vDash \varphi$.
A critical pair $s \approx t~\CO{\varphi}$ is \emph{joinable} if
$s \approx t~\CO{\varphi} \RabRs{*} u \approx v~\CO{\psi}$ and
$u \approx v~\CO{\psi}$ is trivial.
\end{definition}

We revisit \exaref{example2}.

\begin{example}
For the critical pair in \exaref{example2} we obtain
\begin{align*}
\m{h}(\m{g}(x,&\m{1} + \m{1})) \approx \m{h}(\m{g}(y,\m{1} + \m{1})) \\
&\Rca \m{h}(\m{g}(x,v)) \approx \m{h}(\m{g}(y,\m{1} + \m{1}))~
\CO{v = \m{1} + \m{1}} \\
&\Rca \m{h}(\m{g}(x,v)) \approx \m{h}(\m{g}(y,z))~
\CO{v = \m{1} + \m{1} \land z = \m{1} + \m{1}}
\end{align*}
The substitution $\sigma = \SET{v \mapsto \m{2}, z \mapsto \m{2}}$
respects the constraint $v = \m{1} + \m{1} \land z = \m{1} + \m{1}$
but does not equate $\m{h}(\m{g}(x,v))$ and $\m{h}(\m{g}(y,z))$.
\end{example}

The converse of \corref{jcp} also holds, but note that in contrast to
TRSs, joinability of critical pairs is not a decidable criterion for
terminating LCTRSs, due to the undecidable triviality condition.
Moreover, for the converse to hold, it is essental that critical pairs
contain the trivial equations $\psi$ in \defref{ccp}.

\begin{example}
\label{exa:new}
Consider the LCTRS $\xR$ consisting of the rules
\begin{align*}
\m{f}(x) &\R \m{g}(y) & \m{g}(y) &\R \m{a}~\CO{y = y}
\end{align*}
which admits the critical pair
$\m{g}(y) \approx \m{g}(y')~\CO{y = y \land y' = y'}$
originating from the overlap
$\overlap[\epsilon]{\m{f}(x) \R \m{g}(y)}{\m{f}(x') \R \m{g}(y')}$.
This critical pair is joinable as $y$ and $y'$ are restricted to
values and thus both sides rewrite to $\m{a}$ using the second rule.
As $\xR$ is also terminating, it is confluent by \corref{jcp}. If we
were to drop $\psi$ in \defref{ccp}, we would obtain
the non-joinable critial pair $\m{g}(y) \approx \m{g}(y')$ instead
and wrongly conclude non-confluence.
\end{example}

\section{Main Results}
\label{sec:results}

We start with extending a confluence result of Huet~\cite{H80} for linear
TRSs. Below we write $\Rb[\geqslant p]$ to indicate that the position of
the contracted redex in the step is below position $p$.

\begin{definition}
A critical pair $s \approx t~\CO{\varphi}$ is \emph{strongly closed} if
\begin{enumerate}
\item
$s \approx t~\CO{\varphi}
\mr{\,\Rs_{\geqslant 1}^* \cdot \,\Rs_{\geqslant 2}^=\,}
u \approx v~\CO{\psi}$ for some trivial $u \approx v~\CO{\psi}$, and
\smallskip
\item
$s \approx t~\CO{\varphi}
\mr{\,\Rs_{\geqslant 2}^* \cdot \,\Rs_{\geqslant 1}^=\,}
u \approx v~\CO{\psi}$ for some trivial $u \approx v~\CO{\psi}$.
\end{enumerate}
\end{definition}

A binary relation $\R$ on terms is \emph{strongly confluent} if
$t \Ra[*] \cdot \La[=] u$ for all terms $s$, $t$ and $u$ with
$t \L s \R u$. (By symmetry, also $t \Ra[=] \cdot \La[*] u$ is required.)
Strong confluence is a well-known sufficient condition for
confluence. Huet~\cite{H80} proved that linear TRSs are
strongly confluent if all critical pairs are strongly closed. Below
we extend this result to LCTRSs, using the above definition of
strongly closed constrained critical pairs.

\begin{theorem}
\label{thm:sccp}
A linear LCTRS is strongly confluent if all its critical pairs are
strongly closed.
\end{theorem}

We give full proof details in order to illustrate the complications
caused by constrained rewrite rules. The following result from
\cite{WM18} plays an important role.

\begin{lemma}
\label{lem:key}
Suppose $s \approx t~\CO{\varphi} \Rbs[p] u \approx v~\CO{\psi}$
and $\gamma \vDash \varphi$. If $p \geqslant 1$ then $s\gamma \R u\delta$
and $t\gamma = v\delta$ for some substitution $\delta$ with
$\delta \vDash \psi$. If $p \geqslant 2$ then $s\gamma = u\delta$ and
$t\gamma \R v\delta$ for some substitution $\delta$ with
$\delta \vDash \psi$.
\qed
\end{lemma}

\begin{proof}[of \thmref{sccp}]
Consider an arbitrary local peak
\[
t \L_{p_1\sh\rho_1\sh\sigma_1} s \Rb[p_2\sh\rho_2\sh\sigma_2] u
\]
with rewrite rules $\rho_1\colon \crr{\ell_1}{r_1}{\varphi_1}$ and
$\rho_2\colon \crr{\ell_2}{r_2}{\varphi_2}$ from $\xR \cup \xRca$.
We may assume that $\rho_1$ and $\rho_2$ have no variables in common, and
consequently $\Dom(\sigma_1) \cap \Dom(\sigma_2) = \varnothing$.
We have $s|_{p_1} = \ell_1\sigma_1$, $t = s[r_1\sigma_1]_{p_1}$ and
$\sigma_1 \vDash \varphi_1$. Likewise, $s|_{p_2} = \ell_2\sigma_2$,
$u = s[r_2\sigma_2]_{p_2}$ and $\sigma_2 \vDash \varphi_2$.
If $p_1 \parallel p_2$ then
\[
t \Rb[p_2\sh\rho_2\sh\sigma_2] t[r_2\sigma_2]_{p_2} = u[r_1\sigma_1]_{p_1}
\L_{p_1\sh\rho_1\sh\sigma_1} u
\]
Hence both $t \Ra[*] \cdot \La[=] u$ and $t \Ra[=] \cdot \La[*] u$.
If $p_1$ and $p_2$ are not parallel then $p_1 \leqslant p_2$ or
$p_2 < p_1$. Without loss of generality, we consider $p_1 \leqslant p_2$.
Let $q = p_2 \backslash p_1$. We do a case analysis on whether or not
$q \in \FPOS(\ell_1)$.
\begin{itemize}
\item
First suppose $q \notin \FPOS(\ell_1)$. Let $q = q_1q_2$ such that
$q_1 \in \VPOS(\ell_1)$ and let $x$ be the variable in $\ell_1$ at
position $q_1$. We have $\ell_2\sigma_2 = x\sigma_1|_{q_2}$ and thus
$\sigma_1(x) \notin \Val$. Define the substitution $\sigma_1'$ as follows:
\[
\sigma_1'(y) = \begin{cases}
x\sigma_1[r_2\sigma_2]_{q_2} &\text{if $y = x$} \\
\sigma_1(y) &\text{otherwise}
\end{cases}
\]
We show $t \Ra[=] s[r_1\sigma_1']_{p_1} \L u$, which yields
$t \Ra[*] \cdot \La[=] u$ and $t \Ra[=] \cdot \La[*] u$.
Since $\xR$ is left-linear, $\ell_1\sigma_1' =
\ell_1\sigma_1[x\sigma_1']_{q_1} =
\ell_1\sigma_1[x\sigma_1[r_2\sigma_2]_{q_2}]_{q_1} =
\ell_1\sigma_1[r_2\sigma_2]_q$
and thus
$u = s[r_2\sigma_2]_{p_2} = s[\ell_1\sigma_1[r_2\sigma_2]_q]_{p_1} =
s[\ell_1\sigma_1']_{p_1}$.
If we can show $\sigma_1' \vDash \rho_1$ then
$u \R s[r_1\sigma_1']_{p_1}$.
Consider an arbitrary variable $y \in \LVar(\rho_1)$. If $y \neq x$ then
$\sigma_1'(y) = \sigma_1(y) \in \Val$ since $\sigma_1 \vDash \rho_1$.
If $y = x$ then $x \in \Var(\varphi)$ since $x \in \Var(\ell_1)$.
However, this contradicts $\sigma_1 \vDash \rho_1$ as
$\sigma_1(x) \notin \Val$.
So $\sigma_1'(y) = \sigma_1(y)$ for all
$y \in \LVar(\rho_1)$ and thus $\sigma_1' \vDash \rho_1$ is an immediate
consequence of $\sigma_1 \vDash \rho_1$.
It remains to show $t \Ra[=] s[r_1\sigma_1']_{p_1}$. If
$x \notin \Var(r_1)$ then $r_1\sigma_1' = r_1\sigma_1$ and thus
$t = s[r_1\sigma_1']_{p_1}$. If $x \in \Var(r_1)$ then there exists
a unique position $q' \in \VPOS(r_1)$ such that $r_1|_{q'} = x$,
due to the right-linearity of $\xR$. Hence
$r_1\sigma_1' = r_1\sigma_1[x\sigma_1[r_2\sigma_2]_{q_2}]_{q'} =
r_1\sigma_1[r_2\sigma_2]_{q'q_2}$.
Since $r_1\sigma_1|_{q'q_2} = \ell_2\sigma_2$ we obtain
\(
t = s[r_1\sigma_1]_{p_1} \Rb[p_1q'q_2\sh\rho_2\sh\sigma_2]
s[r_1\sigma_1']_{p_1}
\)
as desired.
\smallskip
\item
Next suppose $q \in \FPOS(\ell_1)$.
The substitution $\sigma' = \sigma_1 \cup \sigma_2$ satisfies
$\ell_1|_q\sigma' = \ell_1|_q\sigma_1 = \ell_2\sigma_2 = \ell_2\sigma'$
and thus is a unifier of $\ell_1|_q$ and $\ell_2$. Since
$\sigma_1 \vDash \rho_1$ and $\sigma_2 \vDash \rho_2$,
$\sigma'(x) \in \Val$ for all $x \in \LVar(\rho_1) \cup \LVar(\rho_2)$.
Let $\sigma$ be an mgu of $\ell_1|_q$ and $\ell_2$. Since $\sigma$ is at
least as general as $\sigma'$, $\sigma(x) \in \Val \cup \xV$ for all
$x \in \LVar(\rho_1) \cup \LVar(\rho_2)$.
Since $\varphi_1\sigma' = \varphi_1\sigma_1$ and
$\varphi_2\sigma' = \varphi_2\sigma_2$ are valid,
$\varphi_1\sigma \land \varphi_2\sigma$ is satisfiable.
Hence conditions 1, 2, 3 and 4 in \defref{ccp} hold for the triple
$\overlap[q]{\rho_2}{\rho_1}$.
If condition 5 is \emph{not} fulfilled then
$q = \epsilon$ (and thus $p_1 = p_2$), $\rho_2$ and $\rho_1$ are
variants, and $\Var(r_2) \subseteq \Var(\ell_2)$ (and thus also
$\Var(r_1) \subseteq \Var(\ell_1)$). Hence
$\ell_1\sigma_1 = \ell_2\sigma_2$ and
$r_1\sigma_1 = r_2\sigma_2$, and thus $t = u$.
In the remaining case condition 5 holds and hence
$\overlap[q]{\rho_2}{\rho_1}$ is an overlap. By definition,
$\ell_1\sigma[r_2\sigma]_q \approx r_1\sigma~
\CO{\varphi_2\sigma \land \varphi_1\sigma \land \psi\sigma}$
with 
\[
\psi = \bigwedge~\SET{x = x \mid x \in \EVar(\rho_1) \cup \EVar(\rho_2)}
\]
is a critical pair. To simplify the notation, we abbreviate
$\ell_1\sigma[r_2\sigma]_q$ to $s'$, $r_1\sigma$ to $t'$, and
$\varphi_2\sigma \land \varphi_1\sigma \land \psi\sigma$ to $\varphi'$.
Critical pairs are strongly closed by assumption, and thus both
\begin{enumerate}
\medskip
\item
$s' \approx t'~\CO{\varphi'}
\mr{\,\Rs_{\geqslant 1}^* \cdot \,\Rs_{\geqslant 2}^=\,}
u \approx v~\CO{\psi'}$ for some trivial $u \approx v~\CO{\psi'}$, and
\smallskip
\item
$s' \approx t'~\CO{\varphi'}
\mr{\,\Rs_{\geqslant 2}^* \cdot \,\Rs_{\geqslant 1}^=\,}
u \approx v~\CO{\psi'}$ for some trivial $u \approx v~\CO{\psi'}$.
\end{enumerate}
\smallskip
Let $\gamma$ be the substitution such that $\sigma\gamma = \sigma'$. We
claim that $\gamma$ respects $\varphi'$. So let $x \in \Var(\varphi') =
\Var(\varphi_2\sigma \land \varphi_1\sigma \land \psi\sigma)$. We have
\begin{align*}
\LVar(\rho_1) &= \Var(\varphi_1) \cup \EVar(\rho_1) &
\LVar(\rho_2) &= \Var(\varphi_2) \cup \EVar(\rho_2)
\end{align*}
Together with $\Var(\psi) = \EVar(\rho_1) \cup \EVar(\rho_2)$
we obtain
\[
\LVar(\rho_1) \cup \LVar(\rho_2) = \Var(\varphi_1) \cup
\Var(\varphi_2) \cup \Var(\psi)
\]
Since $\sigma'(x) \in \Val$ for all
$x \in \LVar(\rho_1) \cup \LVar(\rho_2)$, we obtain $\gamma(x) \in \Val$
for all $x \in \Var(\varphi')$ and thus $\gamma \vDash \varphi'$.
At this point repeated applications of \lemref{key} to the constrained
rewrite sequence in item 1 yields a substitution
$\delta$ respecting $\psi'$ such that $s'\gamma \Ra[*] u\delta$ and
$t'\gamma = v\delta$. Since $u \approx v~\CO{\psi'}$ is trivial,
$u\delta = v\delta$ and hence
$s'\gamma \Ra[*] \cdot \La[=] t'\gamma$. Likewise,
$s'\gamma \Ra[=] \cdot \La[*] t'\gamma$ is obtained from
item 2. We have
\begin{align*}
s'\gamma &= (\ell_1\sigma[r_2\sigma]_q)\gamma =
\ell_1\sigma'[r_2\sigma']_q = \ell_1\sigma_1[r_2\sigma_2]_q &
t'\gamma &= r_1\sigma' = r_1\sigma_1
\end{align*}
Moreover, $t = s[r_1\sigma_1]_{p_1} = s[t'\gamma]_{p_1}$ and
$u = s[\ell_1\sigma_1[r_2\sigma_2]_q]_{p_1} = s[s'\gamma]_{p_1}$.
Since rewriting is closed under contexts, we obtain
$u \Ra[*] \cdot \La[=] t$ and $u \Ra[=] \cdot \La[*] t$. This completes
the proof.
\qed
\end{itemize}
\end{proof}

\begin{example}
Consider the LCTRS $\xR$ of \exaref{example1} and its critical pairs
in \exaref{example1 ccp}. The critical pair
\[
x \approx \m{max}(y,x)~\CO{x \geq y}
\]
is not trivial, so \thmref{wo} is not applicable and the rule
$\m{max}(x,y) \R \m{max}(y,x)$ precludes the use of \corref{jcp} to infer
confluence. We do have
\[
x \approx \m{max}(y,x)~\CO{x \geq y} ~\xRa[\geqslant 2]~
x \approx x~\CO{x \geq y}
\]
by applying the rule $\m{max}(x,y) \R y ~ \CO{y \geq x}$ and the
resulting
constrained equation $x \approx x~\CO{x \geq y}$ is obviously trivial.
The same reasoning applies to the critical pair
$y \approx \m{max}(y,x)~\CO{y \geq x}$. The first
critical pair $x \approx y~\CO{x \geq y \land y \geq x}$
in \exaref{example1 ccp} is trivial since any (value) substitution
satisfying its constraint $x \geq y \land y \geq x$
equates $x$ and $y$. By symmetry, all critical pairs of $\xR$
are strongly closed. Since $\xR$ is linear, confluence follows from
\thmref{sccp}.
\end{example}

The second main result is the extension of Huet's parallel closedness
condition on critical pairs in left-linear TRSs~\cite{H80} to LCTRSs.
To this end, we first define parallel rewriting for LCTRSs.

\begin{definition}
\label{def:parallel rewriting}
Let $\xR$ be an LCTRS. The relation $\PRbR$ is defined on terms
inductively as follows:
\begin{enumerate}
\item
$x \PRbR x$ for all variables $x$,
\item
$f(\seq{s}) \PRbR f(\seq{t})$ if $s_i \PRbR t_i$ for all
$1 \leqslant i \leqslant n$,
\item
$\ell\sigma \PRbR r\sigma$ with $\ell \R r~\CO{\varphi} \in \xR$ and
$\sigma \vDash \ell \R r~\CO{\varphi}$,
\item
$f(\seq{v}) \PR v$ with $f \in \xFTh \setminus \Val$,
$\seq{v} \in \Val$ and \\ $v = \inter{f(\seq{v})}$.
\end{enumerate}
\end{definition}

We write $\PRb[\geqslant p]$ to indicate that all positions
of contracted redexes in the parallel step are below $p$.
In the next definition we add constraints to parallel rewriting.

\begin{definition}
\label{def:parallel rewriting constraints}
Let $\xR$ be an LCTRS. The relation $\PRbR$ is defined on constrained terms
inductively as follows:
\begin{enumerate}
\item
$x~\CO{\varphi} \PRbR x~\CO{\varphi}$ for all variables $x$,
\item
$f(\seq{s})~\CO{\varphi} \PRbR f(\seq{t})~\CO{\varphi \land \psi}$ if
$s_i~\CO{\varphi} \PRbR t_i~\CO{\varphi \land \psi_i}$ for all
$1 \leqslant i \leqslant n$ and
$\psi = \psi_1 \land \cdots \land \psi_n$,
\item
$\ell\sigma~\CO{\varphi} \PRbR r\sigma~\CO{\varphi}$ with
$\rho\colon \ell \R r~\CO{\omega} \in \xR$, $\sigma(x) \in \Val \cup
\Var(\varphi)$ for all $x \in \LVar(\rho)$, $\varphi$ is satisfiable and
$\varphi \Rightarrow \omega\sigma$ is valid,
\item
$f(\seq{v})~\CO{\varphi} \PR v~\CO{\varphi \land v = f(\seq{v})}$ with
$\seq{v} \in \Val \cup \Var(\varphi)$,
$f \in \xFTh \setminus \Val$
and $v$ is a fresh variable.
\end{enumerate}
Here we assume that different applications to case 4 result in different
fresh variables. The constraint $\psi$ in case 2 collects the assignments
introduced in earlier applications of case 4. (If there are none,
$\psi = \m{true}$ is omitted.) The same holds for $\seq{\psi}$.
We write $\PRs$ for the relation $\sim \cdot \PRbR \cdot \sim$.
\end{definition}

In light of the earlier developments, the following definition is the
obvious adaptation of parallel closedness for LCTRSs.

\begin{definition}
A critical pair $s \approx t~\CO{\varphi}$ is \emph{parallel closed} if
\[
s \approx t~\CO{\varphi} \PRbs[\geqslant 1] u \approx v~\CO{\psi}
\]
for some trivial $u \approx v~\CO{\psi}$.
\end{definition}

Note that the right-hand side $t$ of the constrained equation
$s \approx t~\CO{\varphi}$ may change due to the equivalence
relation $\sim$, cf.\ the statement of \lemref{key}.

\begin{theorem}
\label{thm:pccp}
A left-linear LCTRS is confluent if its critical pairs are parallel
closed.
\end{theorem}

To prove this result, we adapted the formalized proof presented
in \cite{NM16} to the constrained setting. The required changes are
very similar to the ones in the proof of \thmref{sccp}.

\begin{example}
\label{exa:parallel}
Consider the LCTRS $\xR$ with rules
\begin{align*}
\m{f}(x,y) &\R \m{g}(\m{a},y + y)~\CO{y \geq x \land y = \m{1}} &
\m{a} &\R \m{b} \\
\m{h}(\m{f}(x,y)) &\R \m{h}(\m{g}(\m{b},\m{2}))~\CO{x \geq y} &
\m{g}(x,y) &\R \m{g}(y,x)
\end{align*}
The single critical pair
$\m{h}(\m{g}(\m{a},y + y)) \approx \m{h}(\m{g}(\m{b},\m{2}))~
\CO{y \geq x \land y = \m{1} \land x \geq y}$ is parallel closed:
\begin{align*}
&\m{h}(\m{g}(\m{a},y + y)) \approx \m{h}(\m{g}(\m{b},\m{2}))~
\CO{y \geq x \land y = \m{1} \land x \geq y} \medskip \\
&\qquad
\PRb[\geqslant 1]
\m{h}(\m{g}(\m{b},z)) \approx \m{h}(\m{g}(\m{b},\m{2}))~
\CO{y \geq x \land y = \m{1} \land x \geq y \land z = y + y}
\end{align*}
and the obtained equation is trivial.
Hence $\xR$ is confluent by \thmref{pccp}. Note that the earlier
confluence criteria do no apply.
\end{example}

We also consider the extension of Huet's result by Toyama~\cite{T88},
which has a less restricted joinability condition on critical pairs
stemming from overlapping rules at the root position. Such critical pairs
are called \emph{overlays} whereas critical pairs originating from overlaps
$\overlap{\rho_1}{\rho_2}$ with $p > \epsilon$ are called
\emph{inner} critical pairs.

\begin{definition}
\label{def:apc}
An LCTRS $\xR$ is almost parallel-closed if every inner critical pair is
parallel closed and every overlay $s \approx t~\CO{\varphi}$ satisfies
\[
s \approx t~\CO{\varphi}
\,\PRbs[\geqslant 1] \cdot \Rs_{\geqslant 2}^*\,
u \approx v~\CO{\psi}
\]
for some trivial $u \approx v~\CO{\psi}$.
\end{definition}

\begin{theorem}
\label{thm:apccp}
Left-linear almost parallel-closed LCTRSs are confluent.
\end{theorem}

Again, the formalized proof of the corresponding result for plain TRSs
in \cite{NM16} can be adapted to the constrained setting.

\begin{example}
\label{exa:almostpc}
Consider the following variation of the LCTRS $\xR$ in
\exaref{parallel}:
\begin{align*}
\m{f}(x,y) &\R \m{g}(\m{a},y + y)~\CO{y \geq x \land y = \m{1}} &
\m{a} &\R \m{b} \\
\m{f}(x,y) &\R \m{g}(\m{b},\m{2})~\CO{x \geq y} &
\m{g}(x,y) &\R \m{g}(y,x)
\end{align*}
The overlay
$\m{g}(\m{b},\m{2}) \approx \m{g}(\m{a},y + y)~
\CO{x \geq y \land y \geq x \land y = \m{1}}$ is not
parallel closed but one readily confirms that the condition in
\defref{apc} applies.
\end{example}

\section{Automation}
\label{sec:automation}

As it is very inconvenient and tedious to test by hand if an LCTRS
satisfies one of the confluence criteria presented in the preceding
sections, we provide an implementation. The natural choice would be to
extend the existing tool \ctrl~\cite{KN15} because it is currently the
only tool capable of analyzing confluence of LCTRSs.
However, \ctrl is not actively maintained and not very well documented,
so we decided to develop a new tool for the analysis of LCTRSs.
Our tool is called \crest (\textsf{c}onstrained \textsf{re}writing
\textsf{s}of\textsf{t}ware). It is written in Haskell, based on the
Haskell \texttt{term-rewriting}%
\footnote{\url{https://hackage.haskell.org/package/term-rewriting-0.4.0.2}}
library and allows the logics \texttt{QF\_LIA}, \texttt{QF\_NIA},
\texttt{QF\_LRA}.

The input format of {\crest} is described on its website.%
\footnote{\url{http://cl-informatik.uibk.ac.at/software/crest/}}
After parsing the input, \crest checks that the resulting LCTRS is
well-typed. Missing sort information is inferred.
Next it is checked concurrently whether one of the implemented confluence
criteria applies. \crest supports (weak) orthogonality, strong closedness
and (almost) parallel closedness. The tool outputs the computed critical
pairs and a ``proof'' describing how these are closed, based on the
first criterion that reports a YES result.
Below we describe some of the challenges that one
faces when automating the confluence criteria presented in the preceding
sections.

First of all, how can we determine
whether a constrained critical pair or more generally a
constrained equation $s \approx t~\CO{\varphi}$ is trivial? The following
result explains how this can be solved by an SMT solver.

\begin{definition}
\label{def:trivial}
Given a constrained equation $s \approx t~\CO{\varphi}$, the
formula $T(s,t,\varphi)$ is inductively defined as follows:
\begin{align*}
T(s,t,\varphi) = \begin{cases}
\m{true} &\text{if $s = t$} \\
s = t &\text{if $s, t \in \Val \cup \Var(\varphi)$} \\
\displaystyle \bigwedge_{i = 1}^n T(s_i,t_i,\varphi)
&\text{if $s = f(\seq{s})$ and $t = f(\seq{t})$} \\
\m{false} &\text{otherwise}
\end{cases}
\end{align*}
\end{definition}

\begin{lemma}
A constrained equation $s \approx t~\CO{\varphi}$ is trivial if and only
if the formula $\varphi \implies T(s,t,\varphi)$ is valid.
\end{lemma}

\begin{proof}
First suppose $\varphi \implies T(s,t,\varphi)$ is valid. Let $\sigma$ be
a substitution with
$\sigma \vDash \varphi$. Since $\sigma(x) \in \Val$ for all
$x \in \Var(\varphi)$, we can apply $\sigma$ to the formula
$\varphi \implies T(s,t,\varphi)$. We obtain
$\inter{\varphi\sigma} = \top$ from $\sigma \vDash \varphi$. Hence also
$\inter{T(s,t,\varphi)\sigma} = \top$. Since
$T(s,t,\varphi)$ is a conjunction, the final case in the definition of
$T(s,t,\varphi)$ is not used. Hence $\Pos(s) = \Pos(t)$,
$s(p) = t(p)$ for all internal positions $p$ in $s$ and $t$, and
$s|_p\sigma = t|_p\sigma$ for all leaf positions $p$ in $s$ and $t$.
Consequently, $s\sigma = t\sigma$. This concludes the
triviality proof of $s \approx t~\CO{\varphi}$.

For the only if direction, suppose $s \approx t~\CO{\varphi}$ is
trivial. Note that the variables appearing in the formula
$\varphi \implies T(s,t,\varphi)$ are those of $\varphi$.
Let $\sigma$ be an arbitrary assignment such that
$\inter{\varphi\sigma} = \top$. We need to show
$\inter{T(s,t,\varphi)\sigma} = \top$. We can view $\sigma$ as
a substitution with $\sigma(x) \in \Val$ for all
$x \in \Var(\varphi)$. We have
$\sigma \vDash \varphi$ and thus $s\sigma = t\sigma$ by
the triviality of $s \approx t~\CO{\varphi}$. Hence
$T(s,t,\varphi)$ is a conjunction of equations between
values and variables in $\varphi$, which are turned into
identities by $\sigma$. Hence $\inter{T(s,t,\varphi)\sigma} = \top$ as
desired.
\qed
\end{proof}

The second challenge is how to implement rewriting on constrained
equations in particular, how to deal with the equivalence relation
$\sim$ defined in \defref{constrained rewriting}.

\begin{example}
\label{exa:equivalence}
The LCTRS $\xR$
\begin{align*}
\m{f}(x) &\R z ~ \CO{z = \m{3}} &
\m{g}(\m{f}(x)) &\R \m{a} &
\m{g}(\m{3}) &\R \m{a}
\end{align*}
over the integers admits two critical pairs:
\begin{align*}
&z \approx z'~\CO{z = \m{3} \land z' = \m{3}} &
&\m{g}(z) \approx \m{a}~\CO{z = \m{3}}
\end{align*}
The first one is trivial, but to join the second one, an initial
equivalence step is required:
\[
\m{g}(z) \approx \m{a}~\CO{z = \m{3}} \sim
\m{g}(\m{3}) \approx \m{a}~\CO{z = \m{3}} \R
\m{a} \approx \m{a}~\CO{z = \m{3}}
\]
\end{example}

The transformation introduced below avoids having to look for an
initial equivalence step before a rule becomes applicable. 

\begin{definition}
\label{def:transformation}
Let $\xR$ be an LCTRS. Given a term $t \in \xT(\xF,\xV)$,
we replace values in $t$ by fresh variables and return
the modified term together with the constraint that collects the
bindings:
\begin{align*}
\tf(t) &= \begin{cases}
(t,\m{true}) &\text{if $t \in \xV$} \\
(z, z = t) &\text{if $t \in \Val$ and $z$ is a fresh variable} \\
(f(\seq{s}),\varphi_1 \land \cdots \land \varphi_n)
&\text{if $t = f(\seq{t})$ and $\tf(t_i) = (s_i,\varphi_i)$}
\end{cases}
\end{align*}
Applying the transformation $\tf$ to the left-hand sides of the
rules in $\xR$ produces
\begin{align*}
\tf(\xR) &= \SET{\ell' \R r~\CO{\varphi \land \psi} \mid
\text{$\ell \R r~\CO{\varphi} \in \xR$ and $\tf(\ell) = (\ell',\psi)$}}
\end{align*}
\end{definition}

\begin{example}
Applying the transformation $\tf$ to the LCTRS $\xR$ of
\exaref{equivalence} produces the rules
\begin{align*}
\m{f}(x) &\R z ~ \CO{z = \m{3}} &
\m{g}(\m{f}(x)) &\R \m{a} &
\m{g}(z) &\R \m{a} ~ \CO{z = \m{3}}
\end{align*}
The critical pair 
$\m{g}(z) \approx \m{a}~\CO{z = \m{3}}$ can now be joined by an
application of
the modified third rule. Note that the modified rule does not overlap
with the second rule because $z$ may not be instantiated with $\m{f}(x)$.
Hence the modified LCTRS $\tf(\xR)$ is strongly closed and, because it is
linear, also confluent.
\end{example}

In the following we show the correctness of the transformation.
In particular we prove that the initial rewrite relation is preserved.

\begin{lemma}
\label{lem:R vs tr(R)}
The relations
$\Rb[\xR]$ and $\Rb[\tf(\xR)]$ coincide on unconstrained terms.
\end{lemma}

\begin{proof}
Consider $s, t \in \xT(\xF,\xV)$.
Since the transformation $\tf$ does not affect calculation steps,
it suffices to consider rule steps.
First assume $s = C[\ell\sigma] \Rru C[r\sigma] = t$ by applying the
rule $\ell \R r~\CO{\varphi} \in \xR$ and
let $\ell' \R r~\CO{\varphi'} \in \tf(\xR)$ be its transformation. So
$\tf(\ell) = (\ell',\psi)$ and $\varphi' = \varphi \land \psi$.
Define the substitution
\[
\sigma' = \SET{\ell'|_p \mapsto \ell|_p \mid
\text{$(\ell',\psi) = \tf(\ell)$, $p \in \Pos(\ell)$ and
$\ell|_p \in \Val$}}
\]
and let $\tau = \sigma \cup \sigma'$. Since
$\Dom(\sigma) \cap \Dom(\sigma') = \varnothing$ by construction, $\tau$
is well-defined. From
$\sigma \vDash \ell \R r~\CO{\varphi}$ and $\sigma' \vDash \psi$
we immediately obtain
$\tau \vDash \ell' \R r~\CO{\varphi'}$, which yields
$s = C[\ell'\tau] \Rru C[r\tau] = t$ in $\tf(\xR)$.

For the other direction consider $s = C[\ell'\sigma] \Rru C[r'\sigma] = t$
by applying the rule $\ell' \R r'~\CO{\varphi'} \in \tf(\xR)$.
The difference between $\ell'$ and its originating left-hand side
$\ell$ in $\xR$ is that value positions in $\ell$ are occupied by fresh
variables in $\ell'$.
Because $\sigma'$ respects $\varphi' = \varphi \land \psi$,
$\sigma'$ substitutes the required values at these
positions in $\ell$.
As $\sigma \vDash \ell' \R r'~\CO{\varphi'}$,
there exists a rule $\ell \R r~\CO{\varphi}$ which is respected by
$\sigma$ and thus
$s = C[\ell\sigma] \Rru C[r\sigma] = t$ in $\xR$.
\qed
\end{proof}

As the transformation is used in the implementation and rewriting on
constrained terms plays a key role, the following result is needed.
The proof is similar
to the first half of the proof of \lemref{R vs tr(R)} and omitted.

\begin{lemma}
The inclusion ${\RbR} \subseteq {\Rb[\tf(\xR)]}$ holds on constrained
terms.
\end{lemma}

\begin{proof}
Let $C[\ell\sigma]~\CO{\varphi} \Rru C[r\sigma]~\CO{\varphi}$ by
applying the rule $\ell \R r \CO{\psi} \in \xR$
and let $\ell' \R r~\CO{\psi'} \in \tf(\xR)$
such that $\tf(\ell) = (\ell',\psi)$ and $\psi' = \psi \land \alpha$.
Define the substitution
$\sigma' = \SET{\ell'|_p \mapsto \ell|_p \mid
\text{$(\ell',\alpha) \in \tf(\ell)$, $p \in \Pos(\ell)$ and
$\ell|_p \in \Val$}}$
and let $\tau = \sigma \cup \sigma'$, which is well-defined as
$\Dom(\sigma) \cap \Dom(\sigma') = \varnothing$.
The constraint $\varphi$ is satisfiable by assumption and
$\sigma$ satisfies all the conditions for a constrained rewrite step.
The substitution $\sigma'$ maps variables in $\ell'$ to values such
that $\ell'\sigma' = \ell$.
It remains to show that $\varphi \Rightarrow \psi'\tau$ is valid, which
follows from 
$\psi' = \psi \land \alpha$, the validity of
$\varphi \Rightarrow \psi\tau$, and the fact that
$\psi\sigma$ is equivalent to $\psi'\tau$.
Hence $C[\ell\tau]~\CO{\varphi} \Rru C[r\tau]~\CO{\varphi}$
in $\tf(\xR)$.
\qed
\end{proof}

\section{Experimental Results}
\label{sec:experiments}

In order to evaluate our tool we performed some experiments.
As there is no official database of interesting confluence problems
for LCTRSs, we collected several LCTRSs from the literature
and the repository of \ctrl.
The problem files in the latter that contain an equivalence problem
of two functions for
rewriting induction were split into two separate files.
The experiments were performed on an AMD Ryzen 7 PRO 4750U CPU with a
base clock speed of 1.7 GHz, 8 cores and 32 GB of RAM.
The full set of benchmarks consists of 127 problems of which \crest
can prove 90 confluent, 11 result in MAYBE and 26 in a timeout.
With a timeout of 5 seconds \crest{} needs 141.09 seconds to analyze the
set of benchmarks. 
We have tested the implementation with 3 well-known SMT solvers:
\textsf{Z3}, \textsf{Yices} and \textsf{CVC5}.
Among those \textsf{Z3} gives the best performance regarding time
and the handling of non-linear arithmetic.
Hence we use \textsf{Z3} as the default SMT solver in our implementation.
In \tabref{results} we list some interesting systems from this paper
and the relevant literature.
Full details are available from the website of \crest.
We choose 5 as the maximum number of steps in the $\Ra[*]$ parts of the
strongly closed and almost parallel closed criteria.

\begin{table}[tb]
\renewcommand{\arraystretch}{1.25}
\centering
\caption{Specific experimental results.}
\label{tab:results}
\begin{tabular}{@{}lccr@{}}
\toprule
& result & method & time (in ms) \\
\midrule 
\cite[Example~23]{WM18} & Timeout & --- & 10017.70 \\
\cite[Example~23]{WM18} corrected & YES & strongly closed & 103.71 \\
\exaref{example6} & YES & orthogonal & 34.35 \\
\cite[Example~3]{KN13} & YES & weakly orthogonal & 50.87 \\
\exaref{example1} & YES & strongly closed & 115.33 \\
\cite[Example~1]{NM16} & YES & strongly closed & 3806.84 \\
\exaref{parallel} & YES & parallel closed & 38.42\\
\exaref{almostpc} & YES & almost parallel closed & 130.36\\
\bottomrule
\end{tabular}
\end{table}

From \tabref{overlaps} the relative power of each implemented
confluence criterion on our benchmark can be inferred,
i.e., it depicts how many of the 127 problems both methods
can prove confluent.
This illustrates that the relative applicability in theory
(e.g., weakly orthogonal LCTRSs are parallel closed),
is preserved in our implementation.
We conclude this section with an interesting observation discovered by
\crest when testing \cite[Example~23]{WM18}.

\begin{table}[tb]
\renewcommand{\arraystretch}{1.25}
\centering
\caption{Comparison between confluence criteria implemented in
\crest.}
\label{tab:overlaps}
\begin{tabular}{@{}l@{\quad}r@{\quad}r@{\quad}r@{\quad}r@{\quad}r@{}}
\toprule & O & W & S & P & A \\ \midrule
orthogonality (O) & 74 & 74 & 11 & 74 & 74 \\
weak orthogonality (W) & & 78 & 13 & 78 & 78 \\
strongly closed (S) & & & 20 & 16 & 20 \\
parallel closed (P) & & & & 83 & 83 \\
almost parallel closed (A) & & & & & 89 \\
\bottomrule
\end{tabular}
\end{table}

We also tested the applicability of \corref{jcp}, using the tool
\ctrl as a black box for proving termination. Of the 127 problems,
\ctrl claims 102 to be terminating and 67 of those can be shown
locally confluent by \crest{}, where we limit the number of steps in
the joining sequence to 100.
It is interesting to note that all of these problems are
orthogonal, and so proving termination and finding a
joining sequence is not necessary to conclude confluence, on the
current set of problems. Of the remaining 35 problems, \crest{} can
show confluence of 5 of these by almost parallel closedness.

\begin{example}
The LCTRS $\xR$ 
is obtained by
completing a system consisting of four constrained equations:
\begin{alignat*}{4}
&1. \quad & \m{f}(x,y) &\R \m{f}(z,y) + \m{1} ~
\ML[20]{$\CO{x \geq \m{1} \land z = x - \m{1}}$} \\
&2. & \m{f}(x,\m{0}) &\R \m{g}(\m{1},x) ~ \CO{x \leq \m{1}} \\
&3. & \m{g}(\m{0},y) &\R y ~ \CO{x \leq \m{0}} &
&5. \quad & \m{h}(x) &\R \m{g}(\m{1},x) + \m{1} ~ \CO{x \leq \m{1}} \\
&4. & \m{g}(\m{1},\m{1}) &\R \m{g}(\m{1},\m{0}) + \m{1} &
&6. & \m{h}(x) &\R \m{f}(x - \m{1},\m{0}) + \m{2} ~ \CO{x \geq \m{1}}
\end{alignat*}
Calling \crest on $\xR$ results in a timeout.
As a matter of
fact, the LCTRS is not confluent because the critical pair
\[
\m{g}(\m{1},x) + \m{1} \approx \m{f}(x - \m{1},\m{0}) + \m{2}
~\CO{x \leq \m{1} \land x \geq \m{1}}
\]
between rules 5 and 6 is not joinable. Inspecting the steps
in~\cite[Example~23]{WM18} reveals some incorrect applications of the
inference rules of constrained completion, which causes rule 6 to be
wrong. Replacing it with the correct rule
\begin{alignat*}{2}
&6'. \quad & \m{h}(x) \R (\m{f}(z,\m{0}) + \m{1}) + \m{1} ~
[x > \m{1} \land z = x - \m{1}]
\end{alignat*}
causes \crest to report confluence by strong closedness.
\end{example}

\section{Concluding Remarks}
\label{sec:conclusion}

In this paper we presented new confluence criteria for LCTRSs as well
as a new tool in which these criteria have been implemented. We 
clarified the subtleties that arise when analyzing joinability of
critical pairs in LCTRSs and reported experimental results.

For plain rewrite systems many more confluence criteria are known and
implemented in powerful tools that compete in the yearly Confluence
Competition (CoCo).\footnote{\url{http://project-coco.uibk.ac.at/}}
In the near future we will investigate which of these can be lifted to
LCTRSs. We will also advance the creation of a competition category on
confluence of LCTRSs in CoCo.

Our tool \crest has currently no support for termination. Implementing
termination techniques in \crest is of clear interest. The starting
point here are the methods reported in \cite{K13,K16,WM18}.
Many LCTRSs coming from applications are actually
non-confluent.\footnote{Naoki Nishida, personal communication
(February 2023).} So developing more powerful techniques for LCTRSs is
on our agenda as well.

\subsubsection*{Acknowledgments.}

We thank Fabian Mitterwallner for valuable discussions on the presented
topics and our Haskell implementation.
The detailed comments by the reviewers improved the presentation.
Cynthia Kop and Deivid Vale kindly provided us with instructions and a
working implementation of \ctrl.

\appendix

\section{Proof of Theorems~\ref{thm:pccp} and~\ref{thm:apccp}.}

Since \thmref{pccp} is a special case of \thmref{apccp}, we prove the
latter.
Our proof is based on the proof that has been formalized in
Isabelle/HOL for plain TRSs reported in~\cite{NM16}. The following
definitions and results from \cite{NM16} are needed in the proof.

Recall that a context is a term with exactly one hole ($\hole$).
A multihole context is a term with an arbitrary number of holes.
We write $C[\seq{t}]$ for filling the $n$ holes present in the
multihole context $C$ from left to right with the terms $\seq{t}$.
In order to save space we abbreviate $\seq{a}$ to $\overline{a}$.

\begin{definition}[{\rm \cite[Definition~4]{NM16}}]
\label{def:prmultihole}
We write $s \PRb[C\sh\overline{a}] t$ if $s = C[\seq{a}]$ and
$t = C[\seq{b}]$ for some $\seq{b}$ with
$a_i \Rb[\epsilon] b_i$ for all $1 \leqslant i \leqslant n$.
\end{definition}

It is easy to show that
$s \PR t$ if and only if $s \PRb[C\sh\overline{s}] t$ for some
multihole context $C$ and terms $\overline{s}$.
In order to measure the overlap between two parallel steps starting at the
same term, the overlapping redexes are collected in a multiset.

\begin{definition}[{\rm \cite[Definition~5]{NM16}}]
\label{def:olapmeasure}
The overlap between two co-initial parallel rewrite steps is defined by
the following equations
\begin{align*}
\overlaps\big(\rPRb[\hole\sh{a}] s \PRb[\hole\sh b]\big) &~=~ \SET{s} \\
\overlaps\big(\rPRb[C\sh{\seq[c]{a}}] s \PRb[\hole\sh b]\big)
&~=~ \SET{\seq[c]{a}} \\
\overlaps\big(\rPRb[\hole\sh a] s \PRb[C\sh{\seq[c]{b}}]\big)
&~=~ \SET{\seq[c]{b}} \\
\overlaps\big(\rPRb[f(\seq{C})\sh\overline{a}] s
\PRb[f(\seq{D})\sh\overline{b}]\big) &~=~ \bigcup_{i=1}^{n}
\overlaps\big(\rPRb[C_i\sh\overline{a}_i] s_i
\PRb[D_i\sh\overline{b}_i]\big)
\end{align*}
Here $\overline{a}_1,\dots,\overline{a}_n = \overline{a}$ and
$\overline{b}_1,\dots,\overline{b}_n = \overline{b}$ are partitions of
$\overline{a}$ and $\overline{b}$ such that the length of
$\overline{a}_i$ and $\overline{b}_i$ matches the number of holes in
$C_i$ and $D_i$ for $1 \leqslant i \leqslant n$.
\end{definition}

The properties in the next lemma are crucial for the proof of Theorem~4.

\begin{lemma}[{\rm \cite[Lemma~6]{NM16}}]
\label{lem:overlapprops}
For a peak $\rPRb[C\sh\overline{a}] s \PRb[D\sh\overline{b}]$ the following
properties of $\overlaps(\,\cdot\,)$ hold.
\begin{itemize}
\item
If $s = f(\seq{s})$ with $C = f(\seq{C})$ and $D = f(\seq{D})$ then
\(
\overlaps\big(\rPRb[C_i\sh\overline{a}_i] s_i
\PRb[D_i\sh\overline{b}_i]\big) \,\subseteq\,
\overlaps\big(\rPRb[C\sh\overline{a}] s \PRb[D\sh\overline{b}]\big)
\)
for all $1 \leqslant i \leqslant n$.
\smallskip
\item
The overlap is bounded by $\overline{a}$:~
$\SET{\seq[c]{a}} \,\Supertermmuleq\,
\overlaps\big(\rPRb[C\sh\overline{a}] s \PRb[D\sh\overline{b}]\big)$.
\smallskip
\item
The overlap is symmetric:~
$\overlaps\big(\rPRb[D\sh\overline{b}] s \PRb[C\sh\overline{a}]\big) \,=\,
\overlaps\big(\rPRb[C\sh\overline{a}] s \PRb[D\sh\overline{b}]\big)$.
\end{itemize}
\end{lemma}

The following technical lemma is needed in the third case of the proof of
Theorem~4.

\begin{lemma}[{\rm \cite[Lemma~7]{NM16}}]
\label{lem:linearpr}
Let $s$ be a linear term. If $s\sigma \PRb[C\sh\overline{s}] t$ then
either $t = s\tau$ for some substitution
$\tau$ such that $x\sigma \PR x\tau$ for all $x \in \Var(s)$ or there
exist a context $D$, a non-variable subterm $s'$ of $s$,
a rule $\rho\colon \CRR \in \xR \cup \xRca$, a substitution
$\tau \vDash \rho$, and a multihole context $C'$ such that
$s = D[s']$, $s'\sigma = \ell\tau$,
$D\sigma[r\tau] = C'[s_1,\dots,s_{i-1},s_{i+1},\dots,s_n]$ and
$t = C'[t_1,\dots,t_{i-1},t_{i+1},\dots,t_n]$ for some
$1 \leqslant i \leqslant n$.
\end{lemma}

\begin{lemma}
\label{lem:lifting}
Suppose $s \approx t~\CO{\varphi} \PRbs[\geqslant p] u \approx v~\CO{\psi}$
and $\gamma \vDash \varphi$. If $p = 1$ then $s\gamma \PR u\delta$ and
$t\gamma = v\delta$ for some substitution $\delta$ with
$\delta \vDash \psi$. If $p = 2$ then $s\gamma = u\delta$ and
$t\gamma \PR v\delta$ for some substitution $\delta$ with
$\delta \vDash \psi$.
\end{lemma}

\begin{proof}
The statements are obtained by repeated applications of \lemref{key}. The
key observation, which follows from the proof of \cite[Lemma~8]{WM18}, is
that positions are preserved when lifting rewrite steps on
constrained terms to the unconstrained setting.
\qed
\end{proof}

The lemma above also holds for rewrite sequences
$s \approx t~\CO{\varphi} \Rs_{\geqslant p}^*
u \approx v~\CO{\psi}$.

\begin{proof}[of \thmref{apccp}]
Consider an arbitrary peak
\[
t \rPRb[C\sh\overline{a}] s \PRb[D\sh\overline{b}] u
\]
using rewrite rules from $\xR \cup \xRca$. We show
$t \PR \cdot \rPRa[*] u$ and $t \PR^{*} \cdot \rPR u$ by well-founded
induction on the amount of overlap between the co-initial parallel steps
using the order $\Supertermmul$ and
continue with a second induction on the term $s$ using $\Superterm$.
If $s = x \in \xV$ then $t = u = x$.
Suppose $s = f(\seq{s})$. We consider the following four cases.

\begin{enumerate}
\item
If $C = f(\seq{C})$ and $D = f(\seq{D})$ then $t = f(\seq{t})$,
$u = f(\seq{u})$ and we obtain partitions
$\overline{a}_1,\dots,\overline{a}_n = \overline{a}$ and
$\overline{b}_1,\dots,\overline{b}_n = \overline{b}$ of
$\overline{a}$ and $\overline{b}$ with
$t_i \rPRb[C_i\sh\overline{a}_i] s_i \PRb[D_i\sh\overline{b}_i] u_i$
for all $1 \leqslant i \leqslant n$. The amount of overlap for each
argument position is contained in the amount of overlap of the initial
peak according to \lemref{overlapprops}:
\begin{align*}
\overlaps\big(\rPRb[C_i\sh\overline{a}_i] s_i
\PRb[D_i\sh\overline{b}_i]\big) &~\subseteq~
\overlaps\big(\rPRb[C\sh\overline{a}] s \PRb[D\sh\overline{b}]\big)
\intertext{and thus}
\overlaps\big(\rPRb[C\sh\overline{a}] s \PRb[D\sh\overline{b}]\big)
&~\Supertermmuleq~
\overlaps\big(\rPRb[C_i\sh\overline{a}_i] s_i
\PRb[D_i\sh\overline{b}_i]\big)
\end{align*}
The inner induction hypothesis yields terms $\seq{v}, \seq{w}$ such that
$t_i \PR v_i \rPRa[*] u_i$ and $t_i \PRa[*] w_i \rPR u_i$ for
all $1 \leqslant i \leqslant n$. Hence
$t \PR f(\seq{v}) \rPRa[*] u$ and $t \PRa[*] f(\seq{w}) \rPR u$.
\smallskip
\item
If $C = D = \hole$ then both steps are root steps and thus single rewrite
steps. Hence the given peak can be written as
\[
t = r_1\sigma_1 \L_{\epsilon\sh\rho_1\sh\sigma_1} s
\R_{\epsilon\sh\rho_2\sh\sigma_2} r_2\sigma_2 = u
\]
with rewrite rules $\rho_1\colon \crr{\ell_1}{r_1}{\varphi_1}$ and
$\rho_2\colon \crr{\ell_2}{r_2}{\varphi_2}$ from $\xR \cup \xRca$
and substitutions $\sigma_1$ and $\sigma_2$ that respect these rules. We
have $s = \ell_1\sigma_1 = \ell_2\sigma_2$. We may assume that $\rho_1$
and $\rho_2$ have no variables in common, and thus
$\Dom(\sigma_1) \cap \Dom(\sigma_2) = \varnothing$.
From $\ell_1\sigma_1 = \ell_2\sigma_2$ we infer that
$\sigma' = \sigma_1 \cup \sigma_2$ is a unifier of $\ell_1$ and $\ell_2$.
Let $\sigma$ be an mgu of $\ell_1$ and $\ell_2$. Since
$\sigma_1 \vDash \rho_1$ and $\sigma_2 \vDash \rho_2$,
$\sigma'(x) \in \Val$ for all $x \in \LVar(\rho_1) \cup \LVar(\rho_2)$,
and thus $\sigma(x) \in \xV \cup \Val$ for all
$x \in \LVar(\rho_1) \cup \LVar(\rho_2)$.
Since $\varphi_1\sigma_1$ and $\varphi_2\sigma_2$ are valid,
the constraint $\varphi_1\sigma \land \varphi_2\sigma$ is satisfiable.
Hence conditions 1, 2, 3 and 4 in \defref{ccp} hold for
$\overlap[\epsilon]{\rho_1}{\rho_2}$ and thus we obtain a
critical pair, provided condition 5 also holds.
If condition 5 is \emph{not} fulfilled then $\rho_1$ and $\rho_2$ are
variants and $\Var(r_1) \subseteq \Var(\ell_1)$ (and thus also
$\Var(r_2) \subseteq \Var(\ell_2)$). Hence $r_1\sigma_1 = r_2\sigma_2$ and
thus $t = u$.
In the remaining case condition 5 holds and
$\overlap[\epsilon]{\rho_1}{\rho_2}$ does form an overlap. The induced
constrained critical pair is
$r_2\sigma \approx r_1\sigma~
\CO{\varphi_1\sigma \land \varphi_2\sigma \land \psi\sigma}$ with
\[
\psi = \bigwedge~\SET{x = x \mid x \in \EVar(\rho_1) \cup \EVar(\rho_2)}
\]
To simplify the notation, we abbreviate $r_1\sigma$ to $t'$, $r_2\sigma$
to $u'$, and $\varphi_1\sigma \land \varphi_2\sigma \land \psi\sigma$ to
$\varphi'$. By assumption critical pairs are almost parallel closed and
as we have an overlay we have both
\begin{enumerate}
\medskip
\item
$t' \approx u'~\CO{\varphi'} \mr{\PRbs[\geqslant 1]
\cdot \Rs_{\geqslant 2}^*} v \approx w~\CO{\psi'}$
for some trivial $v \approx w~\CO{\psi'}$, and
\smallskip
\item
$t' \approx u'~\CO{\varphi'} \mr{\PRbs[\geqslant 2]
\cdot \Rs_{\geqslant 1}^*} v' \approx w'~\CO{\psi''}$
for some trivial $v' \approx w'~\CO{\psi''}$.
\end{enumerate}
\smallskip
Let $\gamma$ be the substitution such that $\sigma\gamma = \sigma'$.
We claim that $\gamma$ respects $\varphi'$.
So let $x \in \Var(\varphi') =
\Var(\varphi_2\sigma \land \varphi_1\sigma \land \psi\sigma)$. We have
\begin{align*}
\LVar(\rho_1) &= \Var(\varphi_1) \cup \EVar(\rho_1) &
\LVar(\rho_2) &= \Var(\varphi_2) \cup \EVar(\rho_2)
\end{align*}
Together with $\Var(\psi) = \EVar(\rho_1) \cup \EVar(\rho_2)$ we obtain
\[
\LVar(\rho_1) \cup \LVar(\rho_2) = \Var(\varphi_1) \cup
\Var(\varphi_2) \cup \Var(\psi)
\]
Since $\sigma'(x) \in \Val$ for all
$x \in \LVar(\rho_1) \cup \LVar(\rho_2)$, we obtain $\gamma(x) \in \Val$
for all $x \in \Var(\varphi')$ and thus $\gamma \vDash \varphi'$.
At this point applications of \lemref{lifting}
to the constrained rewrite sequence in item 1 yields a substitution
$\delta$ respecting $\psi'$ such that
$t'\gamma \PR v\delta$ and
$u'\gamma \Ra[*] w\delta$.
Since $v \approx w~\CO{\psi'}$ is trivial, $v\delta = w\delta$ and hence
$t'\gamma \PR \cdot \La[*] u'\gamma$. The sequence
$t'\gamma \Ra[*] \cdot \rPR u'\gamma$ is obtained in the same way from
item 2. Hence
\begin{align*}
t'\gamma &= (r_1\sigma)\gamma = r_1\sigma' = r_1\sigma_1 = t &
u'\gamma &= r_2\sigma' = r_2\sigma_2 = u
\end{align*}
an thus $t \PR \cdot \La[*] u$ and $t \Ra[*] \cdot \rPR u$ as desired.
\smallskip
\item
If $C = f(\seq{C})$ and $D = \hole$ then the step to the right is a single
root step and we have $t = f(\seq{t}) \rPRb[C\sh\overline{a}] s =
\ell\sigma \Rb[\epsilon\sh\rho\sh\sigma] r\sigma = u$
with rule $\rho\colon \ell \R r\ [\varphi]$ from $\xR \cup \xRca$.
Since $\ell$ is linear by assumption, we can apply \lemref{linearpr},
which gives rise to the following two cases.
\begin{enumerate}
\item
In the first case $t = \ell\tau$ for some substitution $\tau$ with
$x\sigma \PR x\tau$ for all $x \in \Var(\ell)$. Define
\[
\delta(x) = \begin{cases}
\tau(x) &\text{if $x \in \Var(\ell)$} \\
\sigma(x) &\text{otherwise}
\end{cases}
\]
We have $t = \ell\tau = \ell\delta$ and claim that
$t \Rb[\epsilon] r\delta$. So we need to show that
$\delta \vDash \rho$.
If $x \in \LVar(\CRR)$ then $\sigma(x) \in \Val$ and thus
$x\sigma = x\tau = x\delta$. Hence also $\delta(x) \in \Val$.
Since $\Var(\varphi) \subseteq \LVar(\CRR)$, we obtain
$\gamma\sigma  = \gamma\delta$. Hence $\delta \vDash \rho$ as desired.
Moreover, $u = r\sigma \PR r \delta$ since $x\sigma \PR x\delta$ for all
variables $x \in \Var(r)$. The latter holds because
$x\sigma \PR x\tau = x\delta$ if
$x \in \Var(\ell)$ and $x\sigma = x\delta$ if $x \notin \Var(\ell)$.
\smallskip
\item
In the second case we obtain a context $E$, a non-variable term $\ell''$,
a rule $\rho'\colon \crr{\ell'}{r'}{\varphi'}$, a substitution $\tau$
respecting $\rho'$, and a multihole context $C'$ such that
$\ell = E[\ell'']$, $\ell''\sigma = \ell'\tau$,
$E\sigma[r'\tau] = C'[s_1,\dots,s_{i-1},s_{i+1},\dots,s_c]$ and
$t = C'[t_1,\dots,t_{i-1},t_{i+1},\dots,t_c]$ for some
$1 \leqslant i \leqslant c$.
Since we may assume without loss of
generality that $\rho$ and $\rho'$ are variable disjoint,
from $\ell''\sigma = \ell'\tau$ we infer that
$\sigma' = \sigma \cup \tau$ is a unifier for
$\ell''$ and $\ell'$. Let $\mu$ be an mgu of $\ell''$ and $\ell'$.
We distinguish two cases, depending on $\rho'$.

\medskip

First suppose $\rho' \in \xR$. Since $C \neq \hole$, we obtain a
constrained critical pair
$E\mu[r'\mu] \approx r\mu~\CO{\varphi'\mu \land \varphi\mu \land \psi\mu}$
with
\[
\psi = \bigwedge~\SET{x = x \mid x \in \EVar(\rho) \cup \EVar(\rho')}
\]
To simplify the notation, we abbreviate $E\mu[r'\mu]$ to $t'$, $r\mu$ to
$u'$, and $\varphi'\mu \land \varphi\mu \land \psi\mu$ to $\varphi''$.
By assumption critical pairs are almost parallel closed and, since we
deal with an inner critical pair, we obtain
\begin{align}
t' \approx u'~\CO{\varphi''} \PRbs[\geqslant 1]
v \approx w~\CO{\psi'}
\label{ps}
\end{align}
for some trivial constrained equation $v \approx w~\CO{\psi'}$.
Let $\gamma$ be the substitution such that $\mu\gamma = \sigma'$.
We claim that $\gamma$ respects $\varphi''$. So let
$x \in \Var(\varphi'') = \Var(\varphi'\mu \land \varphi\mu \land \psi\mu)$.
We have
\begin{align*}
\LVar(\rho') &= \Var(\varphi') \cup \EVar(\rho') &
\LVar(\rho) &= \Var(\varphi) \cup \EVar(\rho)
\end{align*}
Together with $\Var(\psi) = \EVar(\rho') \cup \EVar(\rho)$ we obtain
\[
\LVar(\rho') \cup \LVar(\rho) = \Var(\varphi') \cup \Var(\varphi) \cup
\Var(\psi)
\]
Since $\sigma'(x) \in \Val$ for all $x \in \LVar(\rho') \cup \LVar(\rho)$,
we obtain $\gamma(x) \in \Val$ for all $x \in \Var(\varphi'')$ and
thus $\gamma \vDash \varphi''$. We now apply \lemref{lifting} to
\eqref{ps}. This yields a substitution $\delta$ respecting $\psi'$ such
that $t'\gamma \PR v\delta$ and $u'\gamma = w\delta$. Since
$v \approx w~\CO{\psi'}$ is trivial, $v\delta = w\delta$ and hence
$t'\gamma \PR u'\gamma$. From
$t'\gamma = (E\mu[r'\mu])\gamma = E\sigma'[r'\sigma'] = E\sigma[r'\tau]$
and $u'\gamma = (r\mu)\gamma = r\sigma' = r\sigma = u$ we obtain
$E\sigma[r'\tau] \PR u$. Hence we have a new peak
\[
t \rPRb[C'\sh\overline{a}'] E\sigma[r'\tau] \PR u
\]
with $\overline{a}' = a_1,\dots,a_{i-1},a_{i+1},\ldots,a_c$.
We have
\begin{align*}
\overlaps\big(\rPRb[C\sh\overline{a}] s \PRb[\hole\sh\ell\sigma]\big)
= \SET{\seq[c]{a}} &~\Supertermmul~
\SET{a_1,\dots,a_{i-1},a_{i+1},\dots,a_c} \\
&~\Supertermmuleq~ \overlaps\big(\rPRb[C'\sh\overline{a}']
E\sigma[r'\tau] \PR\big)
\end{align*}
by \lemref{overlapprops}. Hence we can apply the outer induction
hypothesis, which yields $t \PR \cdot \La[*] u$ and
$t \Ra[*] \cdot \rPR u$.

\medskip

Next suppose $\rho' \in \xRca$. In this case
$\ell''\sigma = \ell'\tau = f(\seq{s})$ with $f \in \xFTh \setminus \xFTe$,
$\seq{s} \in \Val$ and $r' = \inter{f(\seq{s})}_\xJ$. So
$r'\tau = r'$ and from this point on we can reason as in the preceding
case.
\end{enumerate}
\smallskip
\item
The final case $D = f(\seq{D})$ and $C = \hole$ is obtained from the
previous one by symmetry.
\end{enumerate}
We have completed the proof of the strong confluence of $\PRbR$,
which implies the confluence of $\RbR$.
\qed
\end{proof}

\bibliographystyle{splncs04}
\bibliography{references}

\end{document}